\documentclass[11pt]{article}
\usepackage{fullpage}
\usepackage{xcolor}
\usepackage{url}
\usepackage{graphicx} %
\usepackage{amsmath, amsthm, amssymb, amsfonts}
\usepackage[hidelinks]{hyperref}
\usepackage[capitalize]{cleveref}
\usepackage{braket}
\theoremstyle{plain} %
\newtheorem{theorem}{Theorem}%
\newtheorem{lemma}[theorem]{Lemma}
\usepackage{tikz}
\usepackage{subcaption}

\usetikzlibrary{arrows.meta,positioning,fit,backgrounds,calc,decorations.pathreplacing}
\newcommand{\FigOneWidth}{0.54\linewidth}  %
\newcommand{\FigTwoWidth}{0.42\linewidth}  %

\newcommand{\PanelAWidth}{0.46\linewidth}  %
\newcommand{\PanelBWidth}{0.50\linewidth}  %
\newcommand{\FigScale}{1.0}

\usetikzlibrary{shapes}
\theoremstyle{definition} %
\newtheorem{definition}{Definition}[section]

\theoremstyle{remark} %
\newtheorem*{remark}{Remark} %
\theoremstyle{remarks} %
\newtheorem*{remarks}{Remarks} %

\DeclareMathOperator{\E}{\mathrm{E}}
\title{Optimal compression with quantum retrieval}

\author{Shyam Dhamapurkar\thanks{Dr.~Vishwanath Karad MIT World Peace University, Pune, India. E-mail: shyam.dhamapurkar@mitwpu.edu.in.}, Mohit Garg\thanks{The Institute of Mathematical Sciences, Chennai, India. Homi Bhabha National Institute, Mumbai, India.\\ E-mail: garg@imsc.res.in.}, Manaswi Paraashar\thanks{Indian Institute of Technology, Hyderabad, India. E-mail: manaswi@cse.iith.ac.in.}, Jaikumar Radhakrishnan\thanks{International Center for Theoretical Sciences, Tata Institute of Fundamental Research, Bengaluru, India.\\ E-mail: jaikumar.radhakrishnan@icts.res.in.}}
\date{28 September 2026}

\begin{document}

\maketitle

\begin{abstract}
    We consider the following data compression problem. Given a string $x \in \{0,1\}^m$ of Hamming weight at most $n$, compress it into a shorter string $y \in \{0,1\}^s$ so that any bit $x_i$ of $x$ can be retrieved without any error using at most $t$ quantum queries to the standard oracle encoding of $y$. If queries are allowed to be adaptive we show how optimal compression up to a logarithmic factor can be achieved. If the queries are required to be made non-adaptively, we show schemes whose space is optimal in its dependence on $m$ except for a logarithmic factor, and is at most quadratically worse when compared to the optimum in its dependence on $n$.
\end{abstract}

\section{Introduction}

We consider the following data compression problem. Given a string $x \in \{0,1\}^m$ of Hamming weight at most $n$, compress it into a shorter string $y \in \{0,1\}^s$ so that any bit $x_i$ of $x$ can be recovered without any error using at most $t$ quantum queries to the standard oracle encoding of $y$. 
Note that the problem is interesting only when $n \leq m/2$ (say), for otherwise we may take $y=x$, and match up to constant factors the information-theoretic minimum of $\log {m \choose n}$ which holds for any compression scheme. Let $s_Q(m,n,t)$ be the minimum $s$ such that such a compression is possible. 
We show the following.
\begin{theorem}\label{thm:tProbeAdaptive} For all $m$, $n\leq m/2$ and $t$, we have
$\displaystyle s_Q(m,n,t) = O\left(n^{1-1/t} m^{1/t} (\log({m}/{n})^{1-1/t})\right)$.
\end{theorem}
\begin{remarks}
    \begin{enumerate}
        \item[(a)] This upper bound  may be compared with the following lower bound derived in \cite{DBLP:conf/focs/RadhakrishnanSV,BMRV} for this problem (valid for $t \le \log (m/n)$): 
        \begin{equation}
s_Q(m,n,t) = \Omega\left(t n^{1-1/t} m^{1/t}\right). \label{eq:lb}
\end{equation}
        The upper bound of \cref{thm:tProbeAdaptive} is worse than this lower bound by a factor $\log (\frac{m}{n})/t$ for $t \leq \log (m/n)$.
        \item[(b)] For $t = \Omega(\log (m/n))$, our upper bound reduces to $O(n \log (m/n))$, which matches the lower bound $\log {m \choose n}$, valid for all values of $t$. The same upper bound for $t = \log(m/n)$ was earlier derived in~\cite{Rasmus}.
    \end{enumerate}
\end{remarks}
In our scheme, the compressed string $y$ is probed in superposition $t$ times, sequentially or \emph{adaptively} (see~\cref{fig:adaptive}). Following the classical literature for this problem, we also consider compression schemes where $t$ probes into $y$ are made in parallel or \emph{non-adaptively}. We model such a query algorithm by a circuit where, to determine the $i$-th bit of the input string $x_i$, an initial joint (superposition) state is prepared on $t$ query registers, after which the oracle for the compressed string $y$ is applied to each of the $t$ query registers. Finally, the bit $x_i$ is recovered by performing an appropriate joint measurement on the final state (see \cref{fig:nonadaptive}). 
Let $s_{QN}(m,n,t)$ be the minimum length of the compressed string for which the bits of the input string can be recovered using $t$ parallel queries. 
\begin{theorem}\label{thm:tProbeNonAdaptive}
For all $m$, $n\leq m/2$ and $t$, we have
    $\displaystyle s_{QN}(m,n,t) = O\left(tn^{2(1-1/t)} m^{1/t} (\log ({m}/{n}))^{1-1/t}\right)$.
\end{theorem}
This upper bound has the right dependence on $m$ (see \cref{eq:lb}); the dependence on $n$ is worse than in \cref{thm:tProbeAdaptive} but it is nevertheless a small polynomial in $n$. When $n$ is a constant, the factor $\log (m/n)$ in \cref{thm:tProbeAdaptive,thm:tProbeNonAdaptive} can be avoided, and a compression all the way to $O(t m^{1/t})$ can be established.

\subsection*{Prior work and the significance of our bounds}

The problem of compressing sparse strings has been considered in several works, mainly in the classical setting, often as the problem of static set membership in the bit-probe model~\cite{BMRV,DBLP:conf/soda/AlonF09, BAIG2020543, DBLP:conf/iwoca/BaigKS19, DBLP:conf/walcom/BaigK18, DBLP:conf/walcom/BaigKS19,DBLP:journals/tcs/BaigK20a, dey_et_al:LIPIcs:2020:12696,Garg16,GR15, MR17, NRR13,RRR01}. In the quantum setting, the problem was studied by Radhakrishnan, Sen and Venkatesh~\cite{DBLP:conf/focs/RadhakrishnanSV}, who, in particular, derived the lower bound in \cref{eq:lb}, which was also the best lower bound for classical (zero-error) schemes at that time. However, no classical scheme that matched this lower bound was known at that time; in fact, subsequently it was shown by Garg and Radhakrishnan that two-probe classical schemes require $\Omega(m)$ space if the string $x$ is allowed to have $\Omega(\log m)$ $1$s; thus, two-probe classical schemes provide no savings over the characteristic vector even when the input string $x$ is rather sparse. Dhamapurkar, Pawar and Radhakrishnan~\cite{dhamapurkar_et_al:LIPIcs.ICALP.2022.52} showed, however, that quantum schemes do significantly better: 
$s_Q(m,m^{1/8},2) = O(m^{7/8})$. This established that compression schemes with quantum retrieval are superior to the classical schemes, and provokes the question: Is the bound 
$s_Q(m,n,2) = \Omega(\sqrt{mn})$ implied by \cref{eq:lb} achievable by quantum schemes? More generally, is the lower bound in  \cref{eq:lb} achievable for larger values of $t$ as well? This work shows, surprisingly, that the answer to both questions is essentially, \emph{Yes!} 

\noindent\emph{Classical versus quantum:}
Furthermore, we note that while \cref{eq:lb} is the best available for large $t$ for classical schemes, the  following counterparts~\cite{MR17,Garg16} of \cref{thm:tProbeAdaptive,thm:tProbeNonAdaptive} for classical schemes do not match this lower bound. 
\begin{enumerate}
        \item[(a)] For odd $t\geq 5$,
        \[
            s_N(m,n,t)
            =O\!\left(
            t n^{1-\frac{2}{t-1}}
            m^{\frac{2}{t-1}}
            \log(2m/n)
            \right).
        \]
        \item[(b)] For odd $t\geq 3$, $t\leq \frac{1}{10}\lg\lg m$, and
        $n\leq m^{1-\epsilon}$ for $\epsilon>0$,
        \[
            s(m,n,t)
            =O\!\left(
            \exp(e^{2t})
            m^{\frac{2}{t+1}}
            n^{1-\frac{2}{t+1}}
            \lg m
            \right).
        \]
    \end{enumerate}
The quantum upper bounds obtained in this paper point to a quantum advantage. It is known that classical two-probe schemes can essentially do no better than a characteristic vector when $n\geq \log m$, whereas our two-probe schemes provide significant compression even when $n=m^{1-\epsilon}$ for a constant $\epsilon > 0$. For larger $t$, we have a classical lower bound
\[ s(m,n,t)=\Omega\!\left(
        m^{\frac{1}{t-1}\left(1-\frac{4^t}{n}\right)}
        \right)
\]
valid for all $n\geq 4^t$. Since the dependence on $m$ in our quantum upper bounds is of the form $m^{1/t}$ instead of $m^{1/(t-1)}$, it shows that quantum schemes continue to enjoy some advantage for larger values of $t$ as well. 

\noindent\emph{Adaptive versus non-adaptive:} It is known that classical \emph{non-adaptive} two-probe schemes are powerless: they offer no compression, not even one bit! In contrast, our two-probe quantum non-adaptive scheme allows compression to $O(n\sqrt{m \log (m/n)})$ bits which provides significant savings for a large range of values of $n$.

\noindent\emph{Shallow quantum circuits:} All the query algorithms that we present have the following property: For every probe, the state of the registers is prepared in superposition, a query is made, a unitary is applied to the registers and all the registers are measured in the computational basis. The lower bound applies to circuits that are allowed to maintain the registers in a superposition till the very end when the final measurement is made. It is remarkable that a scheme with shallow quantum circuits (complete measurement between any two queries) essentially matches the lower bound. It is known~\cite{HLS24} such a shallow circuit offers no advantage for other problems, e.g., database search.

\begin{figure}[t]
\centering
\begin{subfigure}[b]{\FigOneWidth}
\centering
\resizebox{\linewidth}{!}{%
\begin{tikzpicture}[
  x=1cm, y=1cm,
  >=Stealth,
  every node/.style={font=\small},
  gate/.style={draw, thick, fill=blue!4, minimum width=1.1cm, minimum height=1.35cm, align=center},
  oracle/.style={draw, thick, fill=orange!8, minimum width=0.85cm, minimum height=0.85cm, align=center},
  meas/.style={draw, thick, fill=green!6, minimum width=1.0cm, minimum height=0.85cm},
  wire/.style={thick},
  cwire/.style={thick, double, double distance=1.5pt},
  slash/.style={thick},
  brace/.style={decorate, decoration={brace, amplitude=6pt, mirror}, thick},
]
 
  \def\ytop{1.9}
  \def\ya{0.9}
  \def\yb{0}
  \def\gycenter{0.45}
  \def\gtopedge{1.125}
 
  \def\xa{1.2}
  \def\xb{4.0}
  \def\xc{6.8}
  \def\xd{9.6}
  \def\xmeas{11.6}
  \def\xend{13.0}
 
  \pgfmathsetmacro{\mA}{(\xa+\xb)/2}
  \pgfmathsetmacro{\mB}{(\xb+\xc)/2}
  \pgfmathsetmacro{\mC}{(\xc+\xd)/2}
 
  \draw[cwire] (-0.6,\ytop) -- (\xend-0.4,\ytop);
  \node[left] at (-0.7,\ytop) {$i$};
 
  \foreach \x in {\xa,\xb,\xc,\xd} {
    \draw[wire] (\x,\ytop) -- (\x,\gtopedge);
    \fill (\x,\ytop) circle (2.2pt);
  }
 
  \draw[wire] (-0.6,\ya) -- (\xa-0.55,\ya);
  \draw[wire] (-0.6,\yb) -- (\xa-0.55,\yb);
  \draw[slash] (-0.15,\ya-0.13) -- (0.05,\ya+0.13);
  \draw[slash] (-0.15,\yb-0.13) -- (0.05,\yb+0.13);
 
  \node[gate] (U1) at (\xa,\gycenter) {$U_1$};
 
  \draw[wire] (\xa+0.55,\ya) -- (\mA-0.425,\ya);
  \draw[slash] ({(\xa+0.55+\mA-0.425)/2-0.1},\ya-0.13) -- ({(\xa+0.55+\mA-0.425)/2+0.1},\ya+0.13);
  \draw[wire] (\mA+0.425,\ya) -- (\xb-0.55,\ya);
  \draw[slash] ({(\mA+0.425+\xb-0.55)/2-0.1},\ya-0.13) -- ({(\mA+0.425+\xb-0.55)/2+0.1},\ya+0.13);
  \draw[wire] (\xa+0.55,\yb) -- (\xb-0.55,\yb);
  \draw[slash] (\mA-0.1,\yb-0.13) -- (\mA+0.1,\yb+0.13);
 
  \node[oracle] (Oy1) at (\mA,\ya) {$O_y$};
  \node[gate] (U2) at (\xb,\gycenter) {$U_2$};
 
  \draw[wire] (\xb+0.55,\ya) -- (\mB-0.425,\ya);
  \draw[slash] ({(\xb+0.55+\mB-0.425)/2-0.1},\ya-0.13) -- ({(\xb+0.55+\mB-0.425)/2+0.1},\ya+0.13);
  \draw[wire] (\mB+0.425,\ya) -- (\xc-0.55,\ya);
  \draw[slash] ({(\mB+0.425+\xc-0.55)/2-0.1},\ya-0.13) -- ({(\mB+0.425+\xc-0.55)/2+0.1},\ya+0.13);
  \draw[wire] (\xb+0.55,\yb) -- (\xc-0.55,\yb);
  \draw[slash] (\mB-0.1,\yb-0.13) -- (\mB+0.1,\yb+0.13);
 
  \node[oracle] (Oy2) at (\mB,\ya) {$O_y$};
  \node[gate] (U3) at (\xc,\gycenter) {$U_3$};
 
  \draw[wire] (\xc+0.55,\ya) -- (\mC-0.425,\ya);
  \draw[slash] ({(\xc+0.55+\mC-0.425)/2-0.1},\ya-0.13) -- ({(\xc+0.55+\mC-0.425)/2+0.1},\ya+0.13);
  \draw[wire] (\mC+0.425,\ya) -- (\xd-0.55,\ya);
  \draw[slash] ({(\mC+0.425+\xd-0.55)/2-0.1},\ya-0.13) -- ({(\mC+0.425+\xd-0.55)/2+0.1},\ya+0.13);
  \draw[wire] (\xc+0.55,\yb) -- (\xd-0.55,\yb);
  \draw[slash] (\mC-0.1,\yb-0.13) -- (\mC+0.1,\yb+0.13);
 
  \node[oracle] (Oy3) at (\mC,\ya) {$O_y$};
  \node[gate] (U4) at (\xd,\gycenter) {$U_4$};
 
  \draw[wire] (\xd+0.55,\ya) -- (\xmeas-0.5,\ya);
  \draw[slash] (\xd+0.85,\ya-0.13) -- (\xd+1.05,\ya+0.13);
  \draw[wire] (\xd+0.55,\yb) -- (\xd+1.1,\yb);
  \draw[slash] (\xd+0.85,\yb-0.13) -- (\xd+1.05,\yb+0.13);
 
  \node[meas] (M) at (\xmeas,\ya) {};
  \draw[thick] (M.south west) ++(0.12,0.12) to[bend left=35] ++(0.55,0.5);
  \draw[thick,-{Stealth[length=2.2mm]}] (M.south west) ++(0.16,0.1) -- ++(0.6,0.62);
 
  \draw[wire] (\xmeas+0.5,\ya) -- (\xend,\ya);
  \node[right] at (\xend,\ya) {$x_i$};
 
  \draw[brace] (\mA-0.425,-0.35) -- (\mC+0.425,-0.35)
    node[midway, yshift=-0.55cm, align=center, font=\footnotesize]
    {3 \emph{adaptive} queries
    };

\end{tikzpicture}%
}
\caption{A quantum circuit corresponding to an adaptive query scheme that takes an index $i \in [m]$ as classical input and makes three probes to the oracle $O_y$. $U_1$, $U_2$, $U_3$ and $U_4$ are unitaries that depend on the input; the final measurement yields $x_i$.}
\label{fig:adaptive}
\end{subfigure}
\hfill
\begin{subfigure}[b]{\FigTwoWidth}
\centering
\resizebox{\linewidth}{!}{%
\begin{tikzpicture}[
  x=1cm, y=1cm,
  >=Stealth,
  every node/.style={font=\small},
  gate/.style={draw, thick, fill=blue!4, minimum width=1.3cm, minimum height=5.2cm, align=center},
  oracle/.style={draw, thick, fill=orange!8, minimum width=0.85cm, minimum height=0.85cm, align=center},
  meas/.style={draw, thick, fill=green!6, minimum width=1.0cm, minimum height=0.85cm},
  wire/.style={thick},
  cwire/.style={thick, double, double distance=1.5pt},
  slash/.style={thick},
]
 
  \def\ytop{5.3}
  \def\yone{4.0}
  \def\ytwo{2.9}
  \def\ythree{1.8}
  \def\yfour{0.7}
  \def\ybypass{-0.5}
  \def\gycenter{1.75}
  \def\gtopedge{4.35}
 
  \def\xa{1.0}
  \def\xb{7.4}
  \def\xmeas{9.4}
  \def\xend{10.8}
 
  \pgfmathsetmacro{\mAB}{(\xa+\xb)/2}
 
  \draw[cwire] (-0.6,\ytop) -- (\xend-0.4,\ytop);
  \node[left] at (-0.7,\ytop) {$i$};
  \foreach \x in {\xa,\xb} {
    \draw[wire] (\x,\ytop) -- (\x,\gtopedge);
    \fill (\x,\ytop) circle (2.2pt);
  }
 
  \foreach \y in {\yone,\ytwo,\ythree,\yfour,\ybypass} {
    \draw[wire] (-0.6,\y) -- (\xa-0.65,\y);
  }
  \draw[slash] (-0.15,\ybypass-0.13) -- (0.05,\ybypass+0.13);
 
  \node[gate] (U1) at (\xa,\gycenter) {$U_1$};
 
  \foreach \y [count=\i] in {\yone,\ytwo,\ythree,\yfour} {
    \draw[wire] (\xa+0.65,\y) -- (\mAB-0.425,\y);
    \draw[slash] ({(\xa+0.65+\mAB-0.425)/2-0.1},\y-0.13) -- ({(\xa+0.65+\mAB-0.425)/2+0.1},\y+0.13);
    \node[oracle] (Oy\i) at (\mAB,\y) {$O_y$};
    \draw[wire] (\mAB+0.425,\y) -- (\xb-0.65,\y);
    \draw[slash] ({(\mAB+0.425+\xb-0.65)/2-0.1},\y-0.13) -- ({(\mAB+0.425+\xb-0.65)/2+0.1},\y+0.13);
  }
 
  \draw[wire] (\xa+0.65,\ybypass) -- (\xb-0.65,\ybypass);
  \draw[slash] (\mAB-0.1,\ybypass-0.13) -- (\mAB+0.1,\ybypass+0.13);
 
  \node[gate] (U2) at (\xb,\gycenter) {$U_2$};
 
  \draw[wire] (\xb+0.65,\yone) -- (\xmeas-0.5,\yone);
 
  \node[meas] (M) at (\xmeas,\yone) {};
  \draw[thick] (M.south west) ++(0.12,0.12) to[bend left=35] ++(0.55,0.5);
  \draw[thick,-{Stealth[length=2.2mm]}] (M.south west) ++(0.16,0.1) -- ++(0.6,0.62);
  \draw[wire] (\xmeas+0.5,\yone) -- (\xend,\yone);
  \node[right] at (\xend,\yone) {$x_i$};
 
  \foreach \y in {\ytwo,\ythree,\yfour,\ybypass} {
    \draw[wire] (\xb+0.65,\y) -- (\xb+1.3,\y);
  }
  \draw[slash] (\xb+0.95,\ybypass-0.13) -- (\xb+1.15,\ybypass+0.13);
 
  \node[fit=(Oy1)(Oy2)(Oy3)(Oy4), draw, dashed, thick, inner sep=6pt, rounded corners] (grp) {};
 
  \draw[dashed] (grp.south) -- (\mAB,-1.3);
  \node[align=center, font=\footnotesize, text width=5.4cm] at (\mAB,-1.75)
    {4 \emph{non-adaptive} queries};
 
\end{tikzpicture}%
}
\caption{A non-adaptive query scheme that makes four probes to the oracle $O_y$, where all probes are fixed in advance and do not depend on any answers.}
\label{fig:nonadaptive}
\end{subfigure}
\end{figure}

\section*{AI statement}
{We used OpenAI's ChatGPT (Sol and Astra) and Anthropic's Claude (Opus 5.5) for mathematical exploration, proof development, literature searches, and copy editing; the figures were generated by an AI agent from handwritten drawings supplied by the authors. Before consulting AI systems, we had obtained the $O(n\sqrt{m})$ two-probe adaptive scheme based on Bernstein--Vazirani and a non-adaptive scheme with the optimal $m^{1/t}$ dependence on $m$. The $\widetilde{O}(\sqrt{mn})$ adaptive bound was developed through iterative discussions with ChatGPT from an idea supplied by the authors, while the ideas underlying the $\widetilde{O}(n\sqrt{m})$ non-adaptive bound---private neighborhoods and their use with Deutsch--Jozsa---were supplied by the authors and ChatGPT assisted in developing the proof; the extension to $t$ probes used a recursive template previously developed by the authors. The manuscript was written by the authors, with AI-assisted copy editing, and the authors retain full responsibility for its correctness, originality, attribution, and integrity.}

\section{Preliminaries}
\newcommand{\wt}{\mathrm{wt}}
We write $[m]$ for $\{1, \ldots, m\}$. For a string $x \in \{0,1\}^m$, its Hamming weight $\wt_H(x)$ is the number of $1$s in it. We write $B(m,n)$ for the set of all strings in $\{0,1\}^m$ with weight at most $n$, that is, $B(m,n)=\{x \in \{0,1\}^m: \wt_H(x) \leq n\}$.

\begin{definition}[Compression schemes with adaptive and non-adaptive retrieval]
An $(m,n,s,t)$-compression scheme $S$ consists of a \emph{storage scheme} $f_S$ and a query circuit $C_S$. The storage scheme is a function 
\[  f_S: B(m,n) \rightarrow \{0,1\}^s.\]
The query circuit $C_S$ is a quantum circuit that consists of the following:
\begin{enumerate}
    \item A special classical register $I$ of $\log_2 m$ bits capable of carrying an index $i \in [m]$;
    \item Arbitrary unitary operations on the quantum registers that are controlled by the index register $I$;
    \item At most $t$ oracle gates of the form $O_y$ for some $y \in \{0,1\}^s$ (note all oracle gates in the circuit are for the same string $y$), which acts as follows: 
    \[ O_y: \ket{j} \rightarrow (-1)^{y_j} \ket{j}.\]
    \item A measurement at the end that yields an outcome in $\{0,1\}$.
\end{enumerate}
The query scheme is non-adaptive (see~\cref{fig:nonadaptive}) if the oracle gates are applied in parallel; alternatively, we may assume that the query circuit is allowed to use a \emph{product oracle} gate of the form $O_y^{\otimes t'}$ (for some $t' \leq t$), which appears exactly once in the circuit. If the query scheme is allowed to use the outcome of an earlier application of $O_y$ to determine the query fed into a later $O_y$, then we say that it is adaptive (see~\cref{fig:adaptive}).
A compression scheme is \emph{valid} if for all $x \in \{0,1\}^m$ of Hamming weight at most $n$ and all $i \in [m]$, the circuit $C$ with the classical register initialized to $i$ and all its oracle gates as $O_{f(x)}$ yields $x_i$ (with probability $1$) as the outcome of the final measurement. 
Let $s_Q(m,n,t)$ be the smallest $s$ such that there exists a valid $(m,n,s,t)$-compression scheme. Let $s_{QN}(m,n,t)$ be the smallest $s$ such that there exists a valid $(m,n,s,t)$-compression scheme whose query scheme is non-adaptive.
\end{definition}

\begin{remark}
We have modelled quantum query algorithms as circuits. In our proofs however, we will not present the circuits explicitly. We will describe them in informal langauge as a sequence of steps, some of which will be well-known quantum algorithms, e.g., the Bernstein-Vazirani algorithm or the Deutsch-Jozsa algorithm. It will be straightforward to see that one can represent these informally stated quantum algorithms as formal quantum circuits.
\end{remark}

\subsection{Expander graphs}

In our schemes, we will use certain bipartite expander graphs. It will be convenient to represent our bipartite graphs as functions of the form
\[ G: L \times D \rightarrow R,\]
where $L$ represents the left vertices of $G$ and $R$ the right vertices of $G$. In these bipartite graphs, the vertices of $L$ will have $|D|$ edges leaving them and will be labeled by distinct elements of $D$. 
Often $L$ and $R$ will be sets of the form $[m]$, $[k]$, etc., so $L$ and $R$ will not be formally disjoint 
, but when referring to a vertex it will be clear from the context whether a vertex of $L$ or a vertex of $R$ is meant. We use $\Gamma_G(v)$ to refer to the neighbors of a vertex $v\in L$, that is,
\[ \Gamma_G(v) = \{ w \in R: \exists j \in D \, G(v,j) =w\}.\]
The existence of expander graphs with appropriate parameters is established in the following lemmas. 
Note that when these graphs are used in the query algorithm, they will be assumed to be hardwired into the underlying circuit (independent of the input), and will not be stored as part of the compressed string. 
\begin{lemma} \label{lm:expander}
For all $n$ and $k$, there is a bipartite graph $G: [k] \times [d] \rightarrow [2n]$ with $d=O(\max\{\log (k/n),1\})$ such that for every set $A \subseteq [k]$ of size at most $n$, we have $|\Gamma_G(A)|\geq |A|$. 
\end{lemma}
\begin{remark}
    In particular, for graphs promised by the above lemma, by Hall's theorem, 
    for every set $A \subseteq [k]$ of size at most $n$, there is a matching in $G$ that saturates $A$. That is, for every such set $A$, we can find labels $(e_a \in [d]: a \in A)$ such that $(G(a,e_a): a \in A)$ is a sequence of $|A|$ distinct elements in $[2n]$.
\end{remark}
\begin{proof}[Proof of \cref{lm:expander}]
For $k \leq 2n$, we may take $d=1$ and set $G(i,1)=i$ for all $i$. Clearly, this bipartite graph has a matching that saturates $[k]$. So, we assume $k\geq 2n$. 

Pick $G(i,j) \in [2n]$ independently for each pair $(i,j)$ in the domain. By the union bound, the probability that $G$ fails to provide the required expansion is bounded by 
\begin{align*}
 \Pr[\exists \ell\leq n, A\in {[k] \choose \ell}, B\in {[2n] \choose \ell}: \Gamma_G(A) \subseteq B]   
   &\leq \sum_{\ell=2}^n {k \choose \ell} {2n \choose \ell} \left( 
    \frac{\ell}{2n}
    \right)^{\ell d}\\
   &\leq \sum_{\ell=2}^n \left(\frac{ek}{\ell}\right)^\ell 
                         \left(\frac{2en}{\ell}\right)^{\ell}
                         \left(\frac{\ell}{2n}\right)^{\ell d}\\
  &\leq \sum_{\ell=2}^n \left(\left(\frac{2e^2 kn}{\ell^2}\right) \left(\frac{\ell}{2n}\right)^{d} \right)^{\ell}\\
  &\leq \sum_{\ell=2}^n \left(\left(\frac{2e^2 k}{n}\right) 
                        \left(\frac{1}{2}\right)^{d-2} \right)^\ell.
\end{align*}
By setting $d = 3 + \log_2(2e^2 k/n)$, we can ensure that the right hand side is strictly less than $1$. It follows that the random $G$ has the required property with positive probability.
\end{proof}
\begin{lemma} \label{lm:expander2}
For all $n$ and $k$, there is a bipartite graph $G: [k] \times [d] \rightarrow [10nd]$ with $d= O(n\max\{\log (k/n), 1\})$ such that 
\begin{enumerate}
    \item[(i)] $d$ is a power of $2$;
    \item[(ii)] for every $a \in [k]$, we have $|\Gamma_G(a)|=d$;
    \item[(iii)] for every set $A \subseteq [k]$ of size at most $n$ and every $a \in A$, we have
    \[ \left|\Gamma_G(a) \setminus \bigcup_{a' \in A\setminus\{a\}} \Gamma_G(a')\right| \geq \frac{d}{2}.\]
\end{enumerate}
\end{lemma}
\begin{proof}[Proof of \cref{lm:expander2}]
If $k\leq 10n$, it is straightforward to build such a graph $G$, where all vertices in $[k]$ are assigned distinct $n$ neighbors. So we may assume that $k>10n$.

Construct a random $G$ as follows. For each $a \in [k]$, pick the $d$ elements $(G(a,\ell): \ell \in [d])$ \emph{without replacement} from $[10dn]$. It is enough to establish that for every $a \in [k]$ and for every $A' \subseteq [k]\setminus \{a\}$ with $n-1$ elements, we have 
\begin{equation}
 \left|\Gamma_G(a) \cap \bigcup_{a' \in A'} \Gamma_G(a')\right| \leq \frac{d}{2}. \label{eq:property}
 \end{equation}
We will use a union bound over the choices of $a$ and $A'$. 
First fix $a$ and $A'$. Fix the random choices for the neighborhoods of the elements $a' \in A'$. Let $N = \bigcup_{a' \in A'} \Gamma_G(a')$; then, $|N| 
\leq nd$.  When the random choices for the neighbors of $a$ are made, $\Gamma_G(a)$ is a uniformly chosen random set from $[10nd]$ with $d$ elements; then, $\E[|\Gamma_G(a) \cap N|] \leq nd^2/(10nd) = d/10$. By Chv\'{a}tal's bound~\cite{Chvatal79}
\[ \Pr\left[\left|\Gamma_G(a) \cap \bigcup_{a' \in A'} \Gamma_G(a')\right| > \frac{d}{2}\right] \leq \exp\left(-2 (2/5)^2 d\right) \leq \exp(-d/4).\]
Taking a union bound over the $n {k \choose n}$ choices for $(a,A')$, we conclude that the probability that $G$ fails to have the desired property is at most
\[
   n{k \choose n} \exp(-d/4)  \leq n\left(\frac{ek}{n}\right)^n \exp(-d/4).
\] 
Set $d$ to be a power of $2$ at least $4 n (\log (ek/n) + 1)$, so that
the right-hand side is less than $1$. Thus, $G$ has the required properties with positive probability. 
\end{proof}

\section{Two-probe schemes}

As mentioned in the Introduction, the new ideas that lead to \cref{thm:tProbeAdaptive,thm:tProbeNonAdaptive} are already contained in the case of two probes. In this section, we describe these two-probe schemes. 
Each of these schemes is built upon one well-known quantum element.
The adaptive scheme uses the Bernstein-Vazirani algorithm to extract a string $z \in \{0,1\}^{\ell}$ from its $2^\ell$-bit Hadamard encoding $h(z) \in \{0,1\}^{2^\ell}$ by making \emph{just one} quantum probe into the oracle $O_{h(z)}$. 
The non-adaptive scheme uses the Deutsch-Jozsa algorithm to tell apart strings $z \in \{0,1\}^{2^\ell}$ that have exactly $2^{\ell-1}$ $1$s from those that have no $1$s at all, by making \emph{just one} quantum probe into the oracle
$O_z$. 

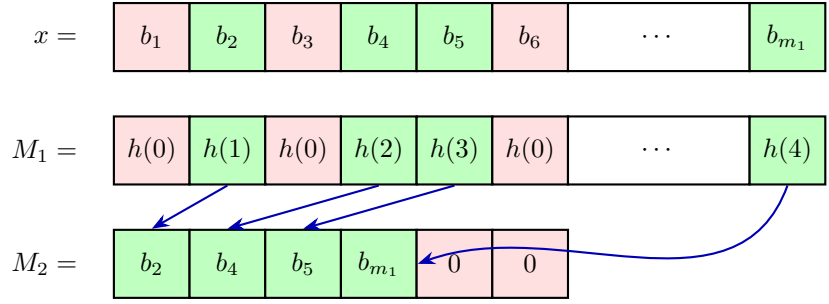
\begin{figure}[t]
\centering
\scalebox{\FigScale}{%
\begin{tikzpicture}[
  >=Stealth,
  every node/.style={font=\small},
  cell/.style={draw, thick, minimum height=0.9cm, inner sep=2pt},
  pink/.style={fill=red!12},
  grn/.style={fill=green!25},
  wht/.style={fill=white},
]
 
  \def\w{1.0}      %
  \def\wdots{2.4}  %
  \pgfmathsetmacro{\ea}{0}
  \pgfmathsetmacro{\eb}{\ea+\w}
  \pgfmathsetmacro{\ec}{\eb+\w}
  \pgfmathsetmacro{\ed}{\ec+\w}
  \pgfmathsetmacro{\ee}{\ed+\w}
  \pgfmathsetmacro{\ef}{\ee+\w}
  \pgfmathsetmacro{\eg}{\ef+\w}
  \pgfmathsetmacro{\eh}{\eg+\wdots}
  \pgfmathsetmacro{\ca}{\ea+\w/2}
  \pgfmathsetmacro{\cb}{\eb+\w/2}
  \pgfmathsetmacro{\cc}{\ec+\w/2}
  \pgfmathsetmacro{\cd}{\ed+\w/2}
  \pgfmathsetmacro{\ceE}{\ee+\w/2}
  \pgfmathsetmacro{\cf}{\ef+\w/2}
  \pgfmathsetmacro{\cg}{\eg+\wdots/2}
  \pgfmathsetmacro{\ch}{\eh+\w/2}
 
  \def\yx{7.0}   %
  \def\ym{5.5}   %
  \def\yn{4.0}   %
 
  \node[left] at (\ea-0.35,\yx) {$x=$};
  \node[left] at (\ea-0.35,\ym) {$M_1=$};
  \node[left] at (\ea-0.35,\yn) {$M_2=$};
 
  \node[cell,pink,minimum width=\w cm] (x1) at (\ca,\yx) {$b_1$};
  \node[cell,grn, minimum width=\w cm] (x2) at (\cb,\yx) {$b_2$};
  \node[cell,pink,minimum width=\w cm] (x3) at (\cc,\yx) {$b_3$};
  \node[cell,grn, minimum width=\w cm] (x4) at (\cd,\yx) {$b_4$};
  \node[cell,grn, minimum width=\w cm] (x5) at (\ceE,\yx) {$b_5$};
  \node[cell,pink,minimum width=\w cm] (x6) at (\cf,\yx) {$b_6$};
  \node[cell,wht, minimum width=\wdots cm] (x7) at (\cg,\yx) {$\cdots$};
  \node[cell,grn, minimum width=\w cm] (x8) at (\ch,\yx) {$b_{m_1}$};
 
  \node[cell,pink,minimum width=\w cm] (m1a) at (\ca,\ym) {$h(0)$};
  \node[cell,grn, minimum width=\w cm] (m1b) at (\cb,\ym) {$h(1)$};
  \node[cell,pink,minimum width=\w cm] (m1c) at (\cc,\ym) {$h(0)$};
  \node[cell,grn, minimum width=\w cm] (m1d) at (\cd,\ym) {$h(2)$};
  \node[cell,grn, minimum width=\w cm] (m1e) at (\ceE,\ym) {$h(3)$};
  \node[cell,pink,minimum width=\w cm] (m1f) at (\cf,\ym) {$h(0)$};
  \node[cell,wht, minimum width=\wdots cm] (m1g) at (\cg,\ym) {$\cdots$};
  \node[cell,grn, minimum width=\w cm] (m1h) at (\ch,\ym) {$h(4)$};
 
  \node[cell,grn, minimum width=\w cm] (n1) at (\ca,\yn) {$b_2$};
  \node[cell,grn, minimum width=\w cm] (n2) at (\cb,\yn) {$b_4$};
  \node[cell,grn, minimum width=\w cm] (n3) at (\cc,\yn) {$b_5$};
  \node[cell,grn, minimum width=\w cm] (n4) at (\cd,\yn) {$b_{m_1}$};
  \node[cell,pink,minimum width=\w cm] (n5) at (\ceE,\yn) {$0$};
  \node[cell,pink,minimum width=\w cm] (n6) at (\cf,\yn) {$0$};
 
  \draw[->, thick, blue!70!black] (m1b.south) -- (n1.north);
  \draw[->, thick, blue!70!black] (m1d.south) -- (n2.north);
  \draw[->, thick, blue!70!black] (m1e.south) -- (n3.north);
  \draw[->, thick, blue!70!black] (m1h.south) to[out=250,in=20] (n4.east);
 
\end{tikzpicture}%
}
\caption{The basic template for the two-probe adaptive scheme}
\label{fig:compaction}
\end{figure}

\subsection{The adaptive two-probe scheme} 
Let $x$ be a string in $B(m,n)$ that we wish to compress.
Partition the sequence $x$ into $m_1$ equal substrings $b_1, b_2, \ldots, b_{m_1}$ of $m/{m_1}$ bits each; we will choose the $m_1$ later in order to obtain a good compression. Since $x$ has weight at most $n$, at most $n$ of the $b_i$ will be non-zero, say $b_{i_1}$, $b_{i_2}$, $\ldots$, $b_{i_r}$, where $i_1 < i_2 < \cdots < i_r$. The idea is to only store these non-zero substrings.
Let $M_2$ be a string of $(m/m_1)n$ bits obtained by concatenating $b_{i_1}$, $b_{i_2}$, $\ldots$, $b_{i_r}$ in that order and, and adding zeroes at the end if $r < n$. When presented an index $i \in [m]$, we need to determine where the substring containing $x_i$ is stored. This will be achieved by making an initial quantum probe. For the first probe, we store for each of the $m_1$ substrings $b_j$ the position in $M_2$, where the that substring is stored, or $0$ if the block was not a part of $M_2$, because it was all zeroes. 
More precisely, for $u=1,2,\ldots, m_1$, let $M_1[u]$ be the Hadamard encoding $h(j)$ for the unique $j$ such that $i_j = u$; 
if no such $i_j$ exists (because $b_u$ is all zeroes) set $M_1[u]=h(0)$. The query algorithm is straightforward. In order to determine $x_i$, it first determines the index $r$, such that $b_r$ contains $x_i$; this can be done using standard deterministic classical computation and requires no probe to the oracle. Then, the algorithm reads the number $j$ from $M_1[r]$ using the Bernstein-Vazirani algorithm; if $j$ is $0$ then the query returns $0$; otherwise, the answer is obtained by reading the appropriate bit of $M_2[j]$. We summarize our observations as follows (see \cref{fig:compaction}).
\begin{lemma}
    Fix $m_1\leq m/2$. There is a valid adaptive two-probe compression scheme that compresses a string $x \in \{0,1\}^m$ into a string of the form $M_1 \circ M_2$, where $|M_1|= 2 m_1 (n+1)$ bits and $|M_2| = (m/m_1) n$. The query algorithm can retrieve any bit of $x$ by making one quantum probe into $M_1$ followed by a classical probe into $M_2$.
\end{lemma}
With this we may set $m_1= \sqrt{m}$ and immediately conclude that $S_Q(m,n,t=2) = O(n\sqrt{m})$. This, falls short of the near optimal $O(\sqrt{mn \log (m/n)})$ bound claimed in \cref{thm:tProbeAdaptive} for $n \gg \log (m/n)$.

In the above scheme, the $i$-th entry of $n+1$ bits in $M_1$ revealed the destination for the substring $M_1$; if the set of possible destinations of each $b_i$ could be made smaller, then we could manage with fewer than $n+1$ bits per entry in $M_1$. 
To get the near optimal bound claimed in \cref{thm:tProbeAdaptive}, we limit the possible destination
using a bipartite expander graph $G: [m_1] \times [d] \rightarrow [2n]$
with the property that for every set $A \subseteq [m_1]$ of $n$ vertices, there is a matching in $G$ that saturates $A$. Such a bipartite graph with $d= O(\log (m_1/n))$ is promised in \cref{lm:expander} (also see the remark immediately after the lemma). Fix such a graph. We find a matching in this bipartite graph that saturates the set $A = \{i_1,i_2,\ldots,i_r\}$ of indices of the non-zero substrings identified above.  Let $\{(i_1, G(i_1,\ell_1)), (i_2, G(i_2,\ell_2)), \ldots
(i_r, G(i_r,\ell_r))\}$ be the matching. For $u=1,2,\ldots,r$, we store $h(\ell_u)$ in $M_1[i_u]$; we store $h(0)$ in the the other positions in $M_1$.
We next describe how we use the matching to arrange the substrings $(b_{i_u}: u=1,2\ldots,r)$ in $M_2$. We imagine $M_2$ to be an array of $2n$ entries, where each entry is string of $m/m_1$ bits. Note that for $u=1,2,\ldots,r$, the numbers $G(i_u,,\ell_u)$ are distinct elements of $[2n]$. We set $M_2[G(i_u,\ell_u)] = b_{i_u}$; that is, we route the non-empty substring $b_{i_u}$ to the position in $M_2$ corresponding the the vertex in $[2n]$ that $i_u$ is matched to. To retrieve $x_i$ from $M_1\circ M_2$, the query algorithm determines the substring $b_j$ that $x_i$ is part of (this needs no probes), determines $\ell_j$ using one quantum probe to $M_1$, and finally reads one bit from the appropriate position of $M_2[G(j,\ell_j)]$. 
We summarize the above observation as follows. 

\begin{lemma} \label{lm:adaptive}
    Fix $m_1\leq m/2$. There is a valid adaptive two-probe compression scheme that compresses a string $x \in \{0,1\}^m$ into a string of the form $M_1 \circ M_2$, where $|M_1|= O(m_1 \log (m_1/n))$ and $|M_2| = 2n (m/m_1)$. 
    The query algorithm can retrieve any bit of $x$ by making one quantum probe into $M_1$ followed by a classical probe into $M_2$.
\end{lemma}
Note that by setting $m_1 = \sqrt{mn/\log (m/n)}$, 
we can conclude from this that $s_Q(m,n,t=2) = O(\sqrt{mn\log (m/n)})$; but for the factor $\sqrt{\log(m/n)}$, this bound is asymptotically optimal. 

\subsection{The non-adaptive two-probe scheme} 

The non-adaptive scheme for $t=2$ bears some similarity to the adaptive schemes discussed above (see~\cref{fig:non-adaptive}). As before, we partition $x$ into $m_1$ substrings $b_1, b_2, \ldots, b_{m_1}$ of $m/{m_1}$ bits each.
The string $M_1$ now has $m_1$ bits; we set $M_1[u]=1$ if $b_{u}$ is non-zero and $0$ otherwise, so one classical probe into $M_1$ reveals if $b_i$ is all zeroes. This time, $M_2$ will be set up for an application of the Deutsch-Jozsa algorithm. The contents of $M_2$ will be based on a bipartite graph $G: [m] \times [d] \to [R]$, with the following property:
\begin{quote}
 For every set $A \subseteq [m]$ with $|A| \leq (m/m_1) n$ elements, for every $a \in A$, we have
 \[ \left|\Gamma_G(a) \setminus \bigcup_{a' \in A\setminus \{a\}} \Gamma_G(a')\right| \geq d/2. \]
\end{quote}
We show below, using~\cref{lm:expander2} how such a bipartite graph with appropriate parameters can be constructed; the bipartite graph will be fixed depending on $m$ and $n$, but independent of the string $x$ that needs to be compressed. For now, let us assume that such bipartite graph can be found. Then, $M_2$ will be a string of length $R$, whose bits are determined as follows.
Let $A$ be the set of indices $a \in [m]$ such that $x[a]$ is part of one of the non-zero strings $b_{i_u}$. Then, $|A|\leq (m/m_1)r \leq (m/m_1) n$. By the above property of the bipartite graph, we may assign to each $a \in A$ a set $B_a$ of neighbors such that $(B_a: a \in A)$ is a sequence of pairwise disjoint sets, each with $d/2$ elements. For all $a \in A$ and $j \in B_a$, we set $M_2[j]=x_a$; the remaining bits of $M_2$ are set to $0$. To determine $x_i$ from $M_1 \circ M_2$, the query algorithm non-adaptively makes one probe each to $M_1$ and $M_2$.
\begin{description}
    \item[Probe into $M_1$:] The algorithm determines (classically) the substring $b_j$ to which $x_i$ belongs; then it reads $e = M_1[j]$. Recall that $e$ indicates whether or not the substring $x_i$ to with $x_i$ belongs is all zeroes.
    \item[Probe into $M_2$:] We regard $M_2$ as a $0$-$1$ labelling of the vertices of $R$. In the bipartite graph $G$, each vertex $i \in [m]$ has a set $\Gamma_G(i) \subseteq R$ of $d$ neighbors. The query algorithm uses the Deutsch-Jozsa algorithm on the labeling $M_2$ restricted to $\Gamma_i$ to obtain a bit $f$. Note that  $M_2$ is designed so that for $i \in A$, all its neighbors are labelled $0$ if $x_i=0$, and  exactly half of them are labelled $0$ if $x_i=1$. 
\end{description}
Using one probe each into the two strings, the query algorithm computes $e$ and $f$ as described above; it then returns $e \wedge f$ as its output.
If $i \in A$, we have $f=x_i$ and $e=1$, and the query algorithm returns the correct answer. For $i \not\in A$, then by definition $x_i=0$; since $e=0$ in this case, the query algorithm's answer is correct. 

\noindent\emph{The expander graph:} Our graph will be of the form $G: [m] \times [d] \rightarrow [R]$. We will view $[m]$ as $[m_1] \times [m/m_1]$ and 
$R$ as $[R'] \times [m/m_1]$; so, the left vertices of the $G$ are of the form 
$(i,j) \in [m_1] \times [m/m_1]$ and the right vertices are of the form $(r,w) \in [R'] \times [m/m_1]$. Let $H: [m_1] \times [d] \rightarrow [10nd]$ be the graph promised by~\cref{lm:expander2}. Then, for $(i,j) \in [m_1] \times [m/m_1]$ and $u \in [d]$, let $G((i,j), u) = (H(i,u),j)$. With the identification of 
$[m]$ with $[m_1] \times [m/m_1]$ and and $[R]$ with $[10nd] \times [m/m_1]$, we obtain a bipartite graph $G: [m] \times [d] \rightarrow [10 n d (m/m_1)]$ with $d = O(n \log (m_1/n))$ that has the required properties. We may view this construction as follows. We write the numbers in $[m]$ in a $m_1 \times (m/m_1)$ array in the natural order, increasing from left to right and top to bottom. We place a copy of $H$ between the indices in each column and a fresh set of $10nd$ vertices; to obtain the required graph $G$, we take the union of all these graphs.  See \cref{fig:non-adaptive}.
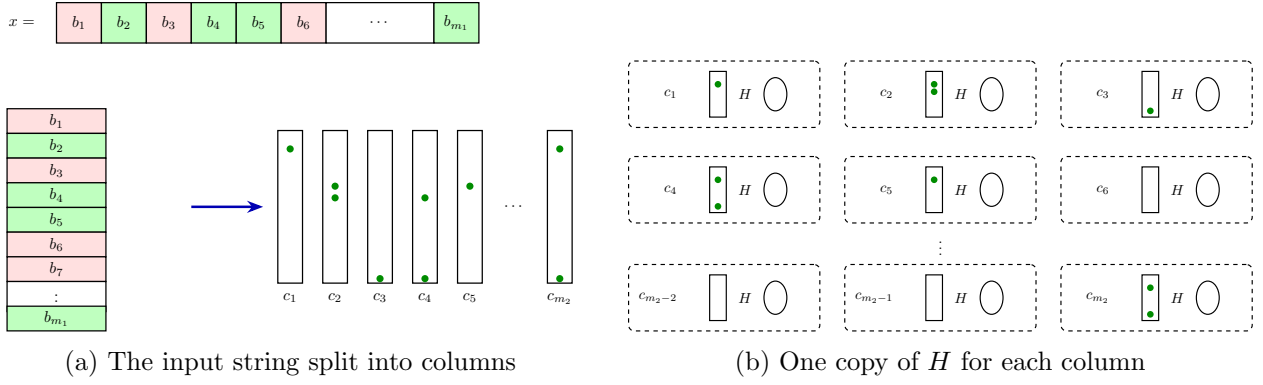
\begin{figure}[t]
\centering
\begin{subfigure}[b]{\PanelAWidth}
\centering
\resizebox{\linewidth}{!}{%
\begin{tikzpicture}[
  >=Stealth,
  every node/.style={font=\small},
  cell/.style={draw, thick, minimum height=0.9cm, inner sep=2pt},
  rowcell/.style={draw, thick, minimum width=2.2cm, minimum height=0.55cm, inner sep=2pt, anchor=north},
  pink/.style={fill=red!12},
  grn/.style={fill=green!25},
  wht/.style={fill=white},
  dot/.style={circle, fill=green!55!black, inner sep=1.6pt},
  bin/.style={draw, thick, minimum width=0.55cm, minimum height=3.4cm},
]
 
  \def\w{1.0}
  \def\wdots{2.4}
  \pgfmathsetmacro{\ea}{0}
  \pgfmathsetmacro{\eb}{\ea+\w}
  \pgfmathsetmacro{\ec}{\eb+\w}
  \pgfmathsetmacro{\ed}{\ec+\w}
  \pgfmathsetmacro{\ee}{\ed+\w}
  \pgfmathsetmacro{\ef}{\ee+\w}
  \pgfmathsetmacro{\eg}{\ef+\w}
  \pgfmathsetmacro{\eh}{\eg+\wdots}
  \pgfmathsetmacro{\ca}{\ea+\w/2}
  \pgfmathsetmacro{\cb}{\eb+\w/2}
  \pgfmathsetmacro{\cc}{\ec+\w/2}
  \pgfmathsetmacro{\cd}{\ed+\w/2}
  \pgfmathsetmacro{\ceE}{\ee+\w/2}
  \pgfmathsetmacro{\cf}{\ef+\w/2}
  \pgfmathsetmacro{\cg}{\eg+\wdots/2}
  \pgfmathsetmacro{\ch}{\eh+\w/2}
  \def\yx{11.3}
 
  \node[left] at (\ea-0.35,\yx) {$x=$};
  \node[cell,pink,minimum width=\w cm] at (\ca,\yx) {$b_1$};
  \node[cell,grn, minimum width=\w cm] at (\cb,\yx) {$b_2$};
  \node[cell,pink,minimum width=\w cm] at (\cc,\yx) {$b_3$};
  \node[cell,grn, minimum width=\w cm] at (\cd,\yx) {$b_4$};
  \node[cell,grn, minimum width=\w cm] at (\ceE,\yx) {$b_5$};
  \node[cell,pink,minimum width=\w cm] at (\cf,\yx) {$b_6$};
  \node[cell,wht, minimum width=\wdots cm] at (\cg,\yx) {$\cdots$};
  \node[cell,grn, minimum width=\w cm] at (\ch,\yx) {$b_{m_1}$};
 
  \def\yc{9.4}
  \node[rowcell,pink] (r1) at (0,\yc)        {$b_1$};
  \node[rowcell,grn]  (r2) at (0,\yc-0.55)   {$b_2$};
  \node[rowcell,pink] (r3) at (0,\yc-1.10)   {$b_3$};
  \node[rowcell,grn]  (r4) at (0,\yc-1.65)   {$b_4$};
  \node[rowcell,grn]  (r5) at (0,\yc-2.20)   {$b_5$};
  \node[rowcell,pink] (r6) at (0,\yc-2.75)   {$b_6$};
  \node[rowcell,pink] (r7) at (0,\yc-3.30)   {$b_7$};
  \node[rowcell,wht]  (r8) at (0,\yc-3.85)   {$\vdots$};
  \node[rowcell,grn]  (r9) at (0,\yc-4.40)   {$b_{m_1}$};
 
  \pgfmathsetmacro{\ymid}{\yc-2.2}
  \draw[->, ultra thick, blue!70!black] (3.0,\ymid) -- (4.6,\ymid);
 
  \def\dtop{1.29}
  \def\dmid{0.20}
  \def\dbot{-1.60}
 
  \pgfmathsetmacro{\bx}{5.2}
  \node[bin] (bin0) at (\bx+0.0,\ymid) {}; \node[below=2pt of bin0] {$c_1$};
    \node[dot] at (\bx+0.0,{\ymid+\dtop}) {};
  \node[bin] (bin1) at (\bx+1.0,\ymid) {}; \node[below=2pt of bin1] {$c_2$};
    \node[dot] at (\bx+1.0,{\ymid+0.46}) {};
    \node[dot] at (\bx+1.0,{\ymid+\dmid}) {};
  \node[bin] (bin2) at (\bx+2.0,\ymid) {}; \node[below=2pt of bin2] {$c_3$};
    \node[dot] at (\bx+2.0,{\ymid+\dbot}) {};
  \node[bin] (bin3) at (\bx+3.0,\ymid) {}; \node[below=2pt of bin3] {$c_4$};
    \node[dot] at (\bx+3.0,{\ymid+\dmid}) {};
    \node[dot] at (\bx+3.0,{\ymid+\dbot}) {};
  \node[bin] (bin4) at (\bx+4.0,\ymid) {}; \node[below=2pt of bin4] {$c_5$};
    \node[dot] at (\bx+4.0,{\ymid+0.46}) {};
 
  \node at (\bx+5.0,\ymid) {$\cdots$};
 
  \node[bin] (binlast) at (\bx+6.0,\ymid) {}; \node[below=2pt of binlast] {$c_{m_2}$};
    \node[dot] at (\bx+6.0,{\ymid+\dtop}) {};
    \node[dot] at (\bx+6.0,{\ymid+\dbot}) {};
 
\end{tikzpicture}%
}
\caption{The input string split into columns}
\label{fig:hash-buckets}
\end{subfigure}
\hfill
\begin{subfigure}[b]{\PanelBWidth}
\centering
\resizebox{\linewidth}{!}{%
\begin{tikzpicture}[
  >=Stealth,
  every node/.style={font=\small},
  dot/.style={circle, fill=green!55!black, inner sep=1.6pt},
  gadget/.style={draw, dashed, thick, rounded corners, inner sep=8pt},
]
 
  \def\ygA{3.0}
  \def\ygB{0.7}
  \def\ygC{-1.9}
  \def\gtop{0.24}
  \def\gmid{0.06}
  \def\gbot{-0.40}
 
  \newcommand{\gadgetbox}[3]{%
    \node[gadget, minimum width=4.6cm, minimum height=1.6cm] (#3) at (#1,#2) {};
    \node[draw, thick, minimum width=0.4cm, minimum height=1.1cm] (#3box) at ($(#3.center)+(-0.15,0)$) {};
    \node[left=6pt of #3box] {};
    \node[right=4pt of #3box] {$H$};
    \node[draw, thick, ellipse, minimum width=0.55cm, minimum height=0.8cm, right=25pt of #3box] {};
  }
 
  \gadgetbox{0.2}{\ygA}{g1}
  \node[anchor=east] at ($(g1.center)+(-1.0,0)$) {$c_1$};
  \node[dot] at ($(g1box.center)+(0,\gtop)$) {};
 
  \gadgetbox{5.4}{\ygA}{g2}
  \node[anchor=east] at ($(g2.center)+(-1.0,0)$) {$c_2$};
  \node[dot] at ($(g2box.center)+(0,\gtop)$) {};
  \node[dot] at ($(g2box.center)+(0,\gmid)$) {};
 
  \gadgetbox{10.6}{\ygA}{g3}
  \node[anchor=east] at ($(g3.center)+(-1.0,0)$) {$c_3$};
  \node[dot] at ($(g3box.center)+(0,\gbot)$) {};
 
  \gadgetbox{0.2}{\ygB}{g4}
  \node[anchor=east] at ($(g4.center)+(-1.0,0)$) {$c_4$};
  \node[dot] at ($(g4box.center)+(0,\gtop)$) {};
  \node[dot] at ($(g4box.center)+(0,\gbot)$) {};
 
  \gadgetbox{5.4}{\ygB}{g5}
  \node[anchor=east] at ($(g5.center)+(-1.0,0)$) {$c_5$};
  \node[dot] at ($(g5box.center)+(0,\gtop)$) {};
 
  \gadgetbox{10.6}{\ygB}{g6}
  \node[anchor=east] at ($(g6.center)+(-1.0,0)$) {$c_6$};
 
  \node at (5.4,{(\ygB+\ygC)/2}) {$\vdots$};
 
  \gadgetbox{0.2}{\ygC}{g7}
  \node[anchor=east] at ($(g7.center)+(-1.0,0)$) {$c_{m_2-2}$};
 
  \gadgetbox{5.4}{\ygC}{g8}
  \node[anchor=east] at ($(g8.center)+(-1.0,0)$) {$c_{m_2-1}$};
 
  \gadgetbox{10.6}{\ygC}{g9}
  \node[anchor=east] at ($(g9.center)+(-1.0,0)$) {$c_{m_2}$};
  \node[dot] at ($(g9box.center)+(0,\gtop)$) {};
  \node[dot] at ($(g9box.center)+(0,\gbot)$) {};
 
\end{tikzpicture}%
}
\caption{One copy of $H$ for each column}
\label{fig:non-adaptive-expander}
\end{subfigure}
\caption{The template for the non-adaptive scheme}
\label{fig:non-adaptive}
\end{figure}
We summarize the above observations.
\begin{lemma} \label{lm:non-adaptive}
    Fix $m_1\leq m/2$. There is a valid non-adaptive two-probe compression scheme that compresses a string $x \in \{0,1\}^m$ into a string of the form $M_1 \circ M_2$, where $|M_1|= m_1$  and $|M_2| = O(n^2 (m/m_1) \log (m_1/n))$. 
    The query algorithm can retrieve any bit of $x$ by making one classical probe into $M_1$ followed by a quantum probe into $M_2$.
\end{lemma}
We may set $m_1=n\sqrt{m \log m}$, and conclude that $s_{QN}(m,n,t=2) = O(n \sqrt{m \log (m_1/n)})$.

\section{The general schemes via recursion}

We observe that a $(m,n,t=2,s)$-compression scheme whose query algorithm makes at least one classical probe into a memory that is always assigned a string of Hamming weight at most $n$, can be naturally lifted, via a recursive construction, to a $t$-probe scheme for every $t\geq 2$, while retaining the adaptive or non-adaptive nature of the original scheme. For $t>2$, the scheme proceeds exactly as the two-probe scheme, except instead of making the classical probe to read the bit, it retrieves the bit using using a $(t-1)$-probe scheme recursively.

Let illustrate this principle with the two-probe schemes presented in the previous section. 
\begin{proof}[Proof of \cref{thm:tProbeAdaptive}]
Consider the adaptive scheme presented in~\cref{lm:adaptive}. 
The scheme uses $M_1$ and $M_2$, where the classical probe is made into $M_2$. 
Furthermore, $|M_1|= m_1 \log (m_1/n)$ and $|M_2|= 2n (m/m_1)$. We may compress $M_2$ recursively, if we are allowed more probes. In particular, this yields the bound:
\[ s_Q(m,n,t) \leq m_1\log(m_1/n)+s_Q(2nm_2,n,t-1).\]
When solved together with the obvious upper bound $s_Q(m,n,t) \leq m$ (in particular, for the base case $t=1$), we obtain
the bound stated in~\cref{thm:tProbeAdaptive}.

A convenient induction hypothesis is
\[
s_Q(m,n,t)\leq
C\,m^{1/t}n^{1-1/t}\log^{1-1/t}(m/n)
\]
for an absolute constant $C$. Let
\[
r=m/n,\qquad L=\log r.
\]
For $t>1$, choose
\[
m_1=A n\left(\frac{r}{L}\right)^{1/t},
\qquad m_2=m/m_1,
\]
where $A$ is a sufficiently large absolute constant. Then
\[
m_1\log(m_1/n)
=O\!\left(
\frac{1}{t}\,
n r^{1/t}L^{1-1/t}
\right),
\]
assuming $t=O(L)$. Moreover, writing
\[
r'=\frac{2nm_2}{n}
=\frac{2}{A}r^{1-1/t}L^{1/t},
\]
we have
\[
\log r'
=\frac{t-1}{t}L+\frac{1}{t}\log L+O(1).
\]
Applying the induction hypothesis to the recursive term gives
\begin{align*}
s_Q(2nm_2,n,t-1)
&\leq
C n(r')^{1/(t-1)}
       \log^{1-1/(t-1)}r'\\
&\leq
C\left(1-\frac{c}{t}\right)
n r^{1/t}L^{1-1/t}
\end{align*}
for some absolute constant $c>0$, by choosing $A$ sufficiently large.
Thus the first term is absorbed by the slack in the recursive term,
for a sufficiently large absolute constant $C$, and hence
\[
s_Q(m,n,t)
=O\!\left(
m^{1/t}n^{1-1/t}
\log^{1-1/t}(m/n)
\right).
\]
\end{proof}

\begin{proof}[Proof of \cref{thm:tProbeNonAdaptive}]
For the non-adaptive scheme, we use \cref{lm:non-adaptive} and obtain
\[ s_{QN}(m,n,t)\leq s_{QN}(m_1,n,t-1) + O(n^2 m_2\log (m_1/n)),\] 
When solved together with the obvious upper bound $s_{QN}(m,n,t) \leq m$ (in particular, for the base case $t=1$), we obtain
the bound stated in~\cref{thm:tProbeNonAdaptive}. A convenient induction hypothesis is
\[
s_{QN}(m,n,t)\leq
C t\,m^{1/t}n^{2-2/t}\log^{1-1/t}(m/n)
\]
for an absolute constant $C$. For $t>1$, choose
\[
m_1=m^{(t-1)/t}\bigl(n^2\log(m/n)\bigr)^{1/t},
\qquad m_2=m/m_1.
\]
Applying the induction hypothesis to the recursive term, and using
$\log(m_1/n)\leq\log(m/n)$, gives
\begin{align*}
s_{QN}(m,n,t)
&\leq C(t-1)m^{1/t}n^{2-2/t}
       \log^{1-1/t}(m/n)+O\!\left(m^{1/t}n^{2-2/t}
       \log^{1-1/t}(m/n)\right)\\
&\leq Ct\,m^{1/t}n^{2-2/t}
       \log^{1-1/t}(m/n),
\end{align*}
for a sufficiently large absolute constant $C$, which completes the proof of \cref{thm:tProbeNonAdaptive}. 
\end{proof}

\bibliographystyle{plain}   
\bibliography{reference} 
\end{document}